\documentclass[sigconf, nonacm]{acmart}

\newcommand\vldbdoi{XX.XX/XXX.XX}

\newcommand\vldbavailabilityurl{}

\usepackage{mathtools}
\usepackage{comment} 
\usepackage{amsthm} 
\usepackage{xspace}                        

\usepackage{listings}

\usepackage{xcolor}
\newcommand{\colA}[1]{\textcolor{blue}{#1}}   
\newcommand{\colB}[1]{\textcolor{orange}{#1}} 

\usepackage{array}
\usepackage{wasysym}

\newcommand{\colfactor}{0.8} 
\newcolumntype{R}{>{\raggedleft\arraybackslash}p{\colfactor cm}}

\theoremstyle{remark}

\newcommand{\sparagraph}[1]{\vspace{1mm}\noindent {\bf #1}}

\usepackage{xcolor}
\usepackage{colortbl}
\usepackage{balance}

\def\strippercent#1\%{#1}

\newcommand{\xboundcell}[1]{%
    \begingroup
    \edef\raw{#1}%
    \edef\val{\expandafter\strippercent\raw}%
    \ifdim \val pt = 0pt
        \cellcolor{green!35}#1%
    \else\ifdim \val pt < 5pt
        \cellcolor{red!10}#1%
    \else\ifdim \val pt < 10pt
        \cellcolor{red!25}#1%
    \else\ifdim \val pt < 20pt
        \cellcolor{red!40}#1%
    \else
        \cellcolor{red!60}#1%
    \fi\fi\fi\fi
    \endgroup
}

\newcommand{\xbound}{\textsc{xBound}\xspace}
\newcommand{\lpbound}{\textsc{LpBound}}
\newcommand{\safebound}{\textsc{SafeBound}}
\newcommand{\panda}{\texttt{PANDA}}

\newcommand{\stats}{\textsc{STATS-CEB}}

\newcommand{\lpinfty}{\ell_{\infty}}
\newcommand{\linfty}{\ell_{\infty}}
\newcommand{\lninfty}{\ell_{-\infty}}
\newcommand{\ltwo}{\ell_{2}}

\newcommand{\va}{\mathbf{a}}
\newcommand{\vb}{\mathbf{b}}

\newcommand{\Dom}[0]{\texttt{Dom}}
\newcommand{\Attrs}[0]{\texttt{Attrs}}
\newcommand{\Vars}[0]{\texttt{Vars}}

\providecommand{\K}{\mathcal{K}}
\providecommand{\lb}[1]{#1^{-}}

\definecolor{dkgreen}{rgb}{0,0.6,0}
\definecolor{gray}{rgb}{0.5,0.5,0.5}
\definecolor{mauve}{rgb}{0.58,0,0.82}
\definecolor{mygray}{rgb}{0.5, 0.5, 0.5}
\usepackage[dvipsnames]{xcolor}

\newif\ifshowcomments
\showcommentstrue

\definecolor{mscolor}{RGB}{40,120,170}     
\definecolor{tbcolor}{RGB}{150,100,40}     
\definecolor{hzcolor}{RGB}{120,70,150}     
\definecolor{jcrcolor}{RGB}{0,140,120}     
\definecolor{ytcolor}{RGB}{170,60,60}      
\definecolor{akcolor}{RGB}{100,120,40}     
\definecolor{todocolor}{RGB}{200,0,0}      

\ifshowcomments
  \newcommand{\ms}[1]{\textcolor{mscolor}{[MS: #1]}}
  \newcommand{\tb}[1]{\textcolor{tbcolor}{[TB: #1]}}
  \newcommand{\hz}[1]{\textcolor{hzcolor}{[HZ: #1]}}
  \newcommand{\jcr}[1]{\textcolor{jcrcolor}{[JCR: #1]}}
  \newcommand{\yt}[1]{\textcolor{ytcolor}{[YT: #1]}}
  \newcommand{\ak}[1]{\textcolor{akcolor}{[AK: #1]}}
  \newcommand{\todo}[1]{\textcolor{todocolor}{[TODO: #1]}}
\else
  \newcommand{\ms}[1]{}
  \newcommand{\tb}[1]{}
  \newcommand{\hz}[1]{}
  \newcommand{\jcr}[1]{}
  \newcommand{\yt}[1]{}
  \newcommand{\ak}[1]{}
  \newcommand{\todo}[1]{}
\fi

\theoremstyle{plain}
\newtheorem{inequality}{Inequality}

\newcommand{\secref}[1]{\S\ref{#1}}

\usepackage{subcaption}
\usepackage{multirow}

\usepackage[switch]{lineno} 

\begin{document}

\title{
    The Case for Cardinality Lower Bounds
}

\author{Mihail Stoian}
\authornote{Work done at Microsoft GSL, except for \secref{sec:improved-gen-rev-hoelder},
~\secref{subsec:prob-l0},
~\secref{subsec:min-degree},
~\secref{subsec:heavy-partition},
~\secref{subsec:norm-stitching},
\secref{subsubsec:conj-and-disj}, and
~\secref{subsec:pred-and-heavy}.
}
\orcid{0000-0002-8843-3374}
\affiliation{%
  \institution{University of Technology Nuremberg}
  \streetaddress{}
  \city{Nuremberg}
  \country{Germany}
  \postcode{}
}
\email{mihail.stoian@utn.de}

\author{Tiemo Bang}
\affiliation{%
  \institution{Microsoft Gray Systems Lab}
  \city{Barcelona}
  \country{Spain}
}
\email{tiemobang@microsoft.com}

\author{Hangdong Zhao}
\affiliation{%
  \institution{Microsoft Gray Systems Lab}
  \city{Redmond, WA}
  \country{USA}}
\email{hangdongzhao@microsoft.com}

\author{Jesús Camacho-Rodríguez}
\affiliation{%
  \institution{Microsoft Fabric DW}
  \city{Mountain View, CA}
  \country{USA}
}
\email{jesusca@microsoft.com}

\author{Yuanyuan Tian}
\affiliation{%
  \institution{Microsoft Gray Systems Lab}
  \city{Mountain View, CA}
  \country{USA}
}
\email{yuanyuantian@microsoft.com}

\author{Andreas Kipf}
\affiliation{%
  \institution{University of Technology Nuremberg}
  \city{Nuremberg}
  \country{Germany}
}
\email{andreas.kipf@utn.de}

\renewcommand{\shortauthors}{Stoian et al.}

\begin{abstract}
Despite decades of research, cardinality estimation remains the optimizer's Achilles heel, with industrial-strength systems exhibiting a systemic tendency toward underestimation. At cloud scale, this is a severe production vulnerability: in Microsoft's Fabric Data Warehouse (DW), a mere 0.05\% of extreme underestimates account for 95\% of all CPU under-allocation, causing preventable slowdowns for thousands of queries daily. Yet recent theoretical work on provable upper bounds only corrects overestimation,
leaving the more harmful problem of underestimation unaddressed. We argue that closing this gap is an urgent priority for the database community.

As a vital step toward this goal, we introduce \xbound{}, the first theoretical framework for computing provable join size lower bounds.
By clipping the optimizer's estimates from below, \xbound{} offers strict mathematical safety nets demanded by production systems---using only a handful of lightweight base table statistics.
We demonstrate \xbound{}'s practical impact on Fabric DW: on the \textsc{StackOverflow-CEB} benchmark, it corrects 23.6\% of Fabric DW's underestimates, yielding end-to-end query speedups of up to 20.1x, demonstrating that even a first step toward provable lower bounds can deliver meaningful production gains and motivating the community to further pursue this critical, open direction. 
\end{abstract}

\maketitle




\section{Introduction}\label{sec:intro}

\begin{figure}[t]
    \centering
    \includegraphics[width=1.0\linewidth]{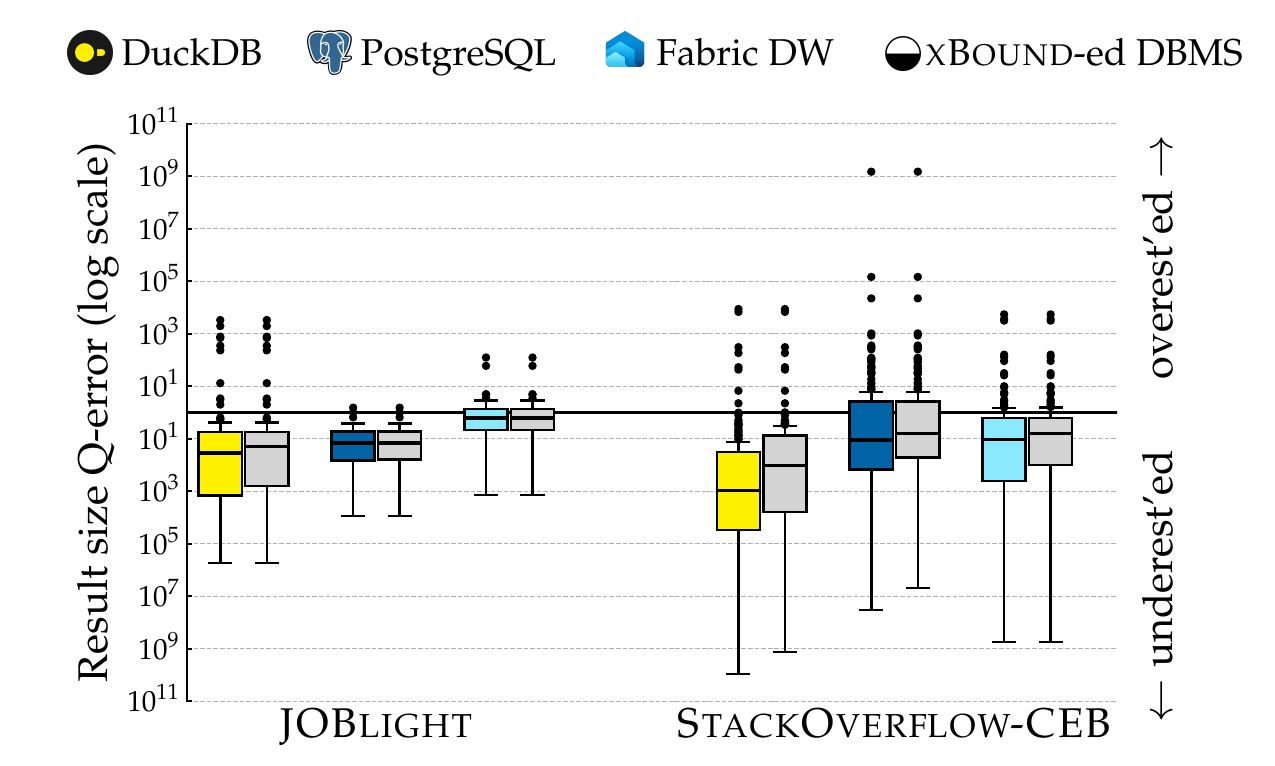}
    \caption{Cardinality underestimation on the \textsc{JOBlight} and \textsc{StackOverflow-CEB} benchmarks. We correct the underestimates by clipping the optimizer's estimates with \xbound{}'s cardinality lower bounds.}
    \label{fig:cardinality-misestimation}
\end{figure}

After five decades of research and development, query optimization remains an ``unsolved'' problem. As pointed out by both practitioners and researchers, cardinality estimation is the \emph{Achilles heel}~\cite{Lohmanblog, Lohmanbtw, job, still-asking} of query optimization. By estimating the size of intermediate results, cardinality estimation is the primary input to the cost model for plan selection. It is equally essential for resource management, dictating memory and CPU allocation. A seminal analysis by Leis et al.~\cite{job} empirically evaluated five open-source and commercial databases, revealing that all routinely produce large estimation errors, with join size estimation being particularly fragile. 

This challenge has spurred a wealth of research on cardinality estimation, spanning sketch-based, sampling-based, to ML-based methods~\cite{CeSurveyFelix, stats-ceb, heinrich-survey}. However, no matter how well a method performs empirically, there is rarely any guarantee on the correctness of the estimation. In response, the theory community has recently proposed provable bounds for cardinality estimation, most notably the recent \lpbound{}~\cite{lpbound}. Yet, all these proposals provide only \emph{upper bounds}. Rather than serving as direct replacements for the cardinality estimator, these upper bounds are best suited for correcting the original optimizer's \emph{overestimates} (by capping the estimates). Crucially, however, the complementary problem, safeguarding against underestimates, remains unaddressed.

\sparagraph{The Still Open Problem of Underestimation.}
The prior focus on upper bounds is fundamentally misaligned with production reality: cardinality underestimation is far more prevalent than overestimation in practice~\cite{job}. When testing DuckDB v1.4, PostgreSQL18, and Fabric Data Warehouse (DW) on \textsc{JOBlight}~\cite{joblight} and \textsc{StackOverflow-CEB}~\cite{stats-ceb}\footnote{That is, \textsc{STATS-CEB} run on the 220 GB StackOverflow dump from \textsc{SQLStorm}~\cite{sqlstorm}.}, all three systems overwhelmingly underestimate result sizes (see Fig.~\ref{fig:cardinality-misestimation}). Moreover, although misestimation in both directions is problematic, underestimation is generally more dangerous. It tricks the optimizer into selecting fragile plans designed for small data, favoring nested-loops over hash joins or broadcast over shuffle joins, causing catastrophic plan regressions and order-of-magnitude slowdowns. At runtime, it risks resource starvation: out-of-memory errors and excessive disk spilling. Overestimation, by contrast, leads to overly conservative plans and overallocated resources, inefficient but seldom catastrophic.

\sparagraph{The Peril of Extreme Underestimation.} In large-scale production environments, extreme underestimation is a systemic vulnerability that severely degrades query performance. This is clearly signaled by CPU resource under-allocation. Fabric DW, for instance, allocates CPU per query stage based on its operators' estimated output cardinality and row size.
Fig.~\ref{fig:dw-disk-spilling} reveals a heavy negative tail in Fabric's CPU estimation: while >99\% of stages face no underestimation (cardinality underestimation does not always lead to CPU underestimation), 1 in 10'000 stages suffers from >70x CPU under-allocation, with extreme outliers starved by over a thousandfold. At Fabric's scale, this causes thousands of queries to suffer severe, preventable slowdowns every day. The database community can no longer afford to ignore this blind spot; it is imperative that both academia and industry actively shift their focus toward mitigating extreme underestimation.

\sparagraph{The Case for Cardinality Lower Bounds.}
As a vital step towards solving this problem, we introduce \xbound, the first theoretical framework for computing provable join size lower bounds using only a handful of lightweight base table statistics. Clipping optimizer estimates with these bounds significantly improves Q-errors, as shown in Fig.~\ref{fig:cardinality-misestimation}, notably reducing Fabric DW's 90th-percentile Q-error of its underestimates on \textsc{StackOverflow-CEB} by 35.8x. Theoretical lower bounds are uniquely suited for adoption in production systems because they provide strict mathematical guarantees against catastrophic worst-case scenarios, unlike empirical methods that may perform well on average but lack necessary safety nets.

Similar to \lpbound{} and sketch-based estimators~\cite{join-sketch, nicolau-sketch}, \xbound the key fact that the size of a join is the inner product of the join-key degree vectors $\va$ and $\vb$ of two relations $A$ and $B$, respectively. Our approach relies on the following observation: Lower bounds require inequalities that bound from below this very inner product; we refer to these as \emph{reverse} inequalities. Fortunately, reverse inequalities only require a handful of statistics: $\lpinfty$---the maximum key frequency, $\ltwo$---the Euclidean norm of the degree vector, and $\lninfty$---the minimum key frequency. On the other hand, they require \emph{positive} input vectors. We will thus first lower bound the number of non-zero entries in the entry-wise product $\va \odot \vb$, enabling their application. 

Results show that \xbound corrects 23.6\% of underestimates on the \textsc{StackOverflow-CEB} benchmark, yielding end-to-end query speedups of up to 20.1x on Fabric DW.

\sparagraph{Contributions.} We summarize our contributions in the following:
\begin{itemize}
  \item[(1)] We make the case for cardinality lower bounds by pointing out severe underestimation in production, and introduce \xbound, the first framework for provable lower bounds on join sizes. In its current form, \xbound{} estimates multi-way joins over the same join key.
  \item[(2)] We extend our framework to obtain lower bounds over filtered table scans, supporting equality and range predicates as well as conjunctions and disjunctions.
  \item[(3)] We empirically evaluate the impact of correcting Fabric DW's underestimates with \xbound{}'s lower bounds on the \textsc{StackOverflow-CEB} benchmark.
\end{itemize}

\sparagraph{Structure.} The rest of the paper is structured as follows: We provide the preliminaries needed in \secref{sec:prelims}, including $\ell_{p}$-norms and inner-product inequalities. Then, in \secref{sec:xbound}, we detail \xbound{}'s components. Subsequently, we extend our framework to support lower bounds on filtered table scans in \secref{sec:base-stats}. We continue with the evaluation of \xbound{}'s lower bounds in \secref{sec:experiments}, and then with a discussion in \secref{sec:discussion}, where we motivate further research on lower bounds. Afterwards, we outline related work in \secref{sec:rel-work} and conclude in \secref{sec:conclusion}.

\begin{figure}[t]
    \centering
    \includegraphics[width=1\linewidth]{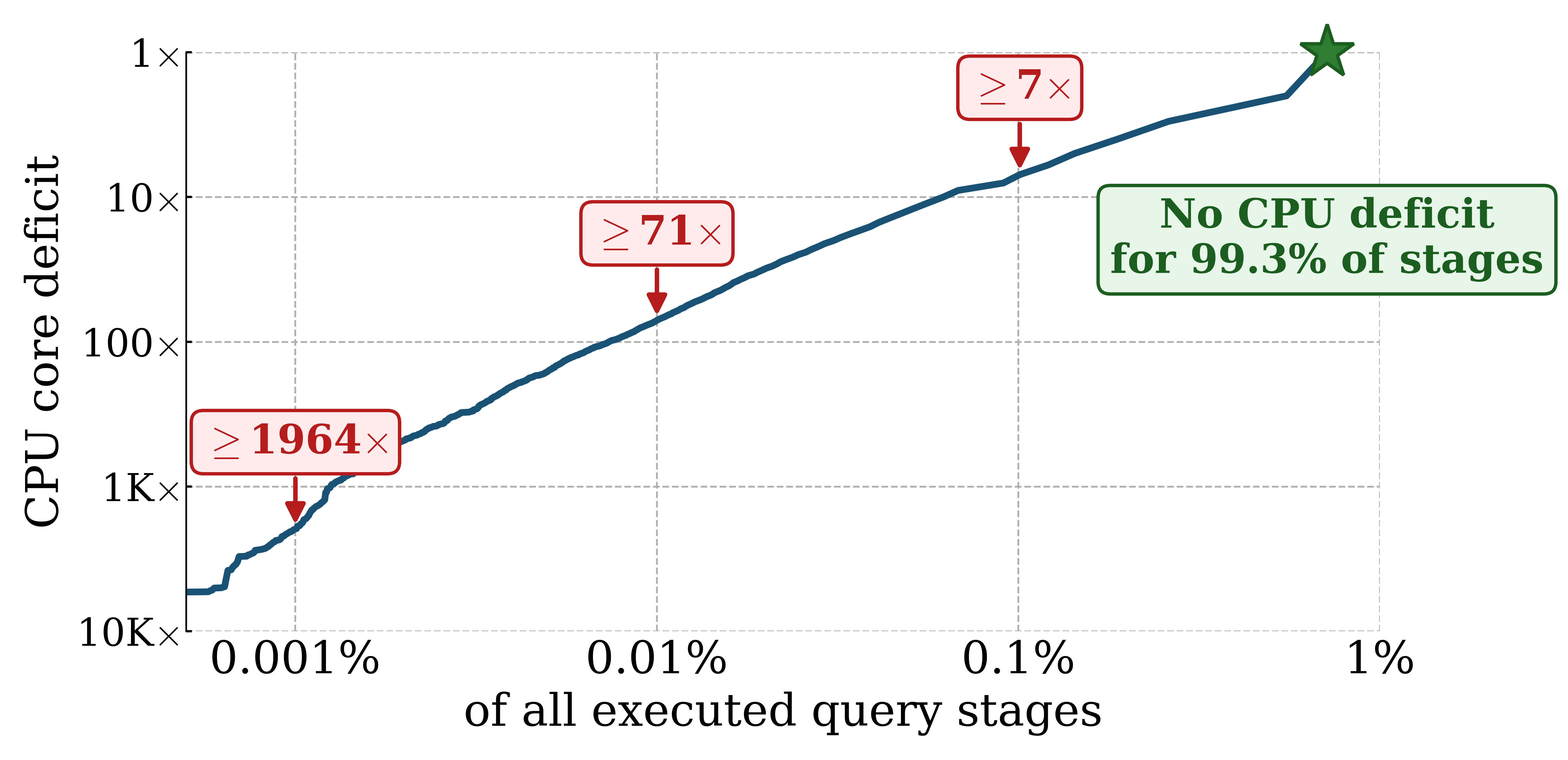}
    \caption{CPU deficit due to cardinality underestimation of query stages executed in production of Fabric Data Warehouse in Feb.~2026 across all cloud regions.}
    \label{fig:dw-disk-spilling}
\end{figure}

\section{Preliminaries}\label{sec:prelims}

\sparagraph{Notation.} For a given relation $R$, denote by $\texttt{Attrs}(R)$ its set of attributes.
Let $|\cdot|$ denote the cardinality of a given set. To denote bounds on the cardinality, we use $|\cdot|^-$ and $|\cdot|^+$ to mean a lower and upper bound, respectively.
We use bold characters for vectors, e.g., $\va$.
To access position $k$ in a vector $\mathbf{a}$, we write $\mathbf{a}[k]$.
For a set of positions $K = \{k_1, \ldots, k_m\}$, we write $\va[K]$ to denote the vector of values at the corresponding positions.
In turn, to index a sequence $d$, we use subscript notation, i.e., $d_k$ (this distinction helps with readability in the proofs).
For integers $a$ and $b$, we define $[a, b] \triangleq \{a, \ldots, b\}$, and, as shorthand, $[n] \triangleq [1, n]$.

\sparagraph{Supported Query Type.} Our framework supports select-project-join queries, as follows:
\begin{lstlisting}[label=l:query-example,language=SQL]
select [projection-list]
from R_1, ..., R_n
where [inner-join-predicates] and [table-predicates];
\end{lstlisting}
Throughout the paper, we will assume that the queries have bag semantics, as it is currently the case for leading database systems.\footnote{We leave it as future work to understand the effect of our framework for set semantics.}
To ensure notation consistency with prior work on join size upper bounds, we closely follow the notation in Ref.~\cite{lpbound}, i.e., we will use the conjunctive query notation instead of SQL: \[
  Q(V_0) = R_1(V_1) \land R_2(V_2) \land \ldots \land R_n(V_n) \land [\text{table-predicates}],
\]
where each $V_i \triangleq \texttt{Attrs}(R_i)$ is the corresponding set of variables of relation $R_i$, and $V_0 \subseteq V_1 \cup \ldots \cup V_n$ are the \texttt{group-by} attributes.
Note that, since we will be working with \emph{full} conjunctive queries only, it holds that $V_0 = V_1\:\cup \ldots \cup\:V_n$, i.e., the query does not have a \texttt{group-by} clause. 
To this end, let $\Vars(Q) \triangleq V_1\:\cup \ldots \cup\:V_n$ denote the set of query variables, and we will write $Q(\ast)$ as a syntactic sugar for a full conjunctive query, i.e., $V_0 = \Vars(Q)$.

Query $Q$ is acyclic if its relations $R_1, \ldots, R_n$ can be placed on the nodes of a tree such that, for each variable $X_i \in \Vars(Q)$,
the set of tree nodes containing $X_i$ forms a connected component. For instance, the three-way joins $J_1$ and $J_2$ are acyclic,\footnote{Even Berge-acyclic, since every two relations share at most one variable.} yet $J_3$ is cyclic:
\begin{flalign}
  J_1(\ast) &= R_1(X, Y) \land R_2(Y, Z) \land R_3(Z, U),\\
  J_2(\ast) &= R_1(X, Y) \land R_2(X, Z) \land R_3(X, U),\\
  J_3(\ast) &= R_1(X, Y) \land R_2(Y, Z) \land R_3(Z, X).
\end{flalign}
In this work, we focus solely on inner and outer joins on single keys. Notably, this is the case for many subexpressions of the well-studied benchmarks, such as JOB~\cite{job, joblight} and \stats{}~\cite{stats-ceb}.

\sparagraph{Join Pool.} We refer to a join pool as a set of attributes that are joined over in the same query, e.g., the \texttt{userid} across all tables of a social network schema. Join pools can be discovered either directly from the schema or, in its absence, dynamically as the query workload is running. This definition will prove helpful when introducing the heavy partition in \secref{subsec:heavy-partition}.

\sparagraph{Degree Sequences \& Vectors.} The degree sequence from $X$ to $Y$ in relation $R$ is the sequence $\deg(Y | X) \triangleq (d_1, \ldots, d_N)$, constructed as follows:
Compute the domain of $X$, $\Dom(R.X) = \{x_1, \ldots, x_N\}$, and denote by $d_i = |\sigma_{X = x_i}(\Pi_{XY}(R))|$ the degree, i.e., the frequency, of $x_i$, and sort the values in \emph{ascending} order $d_1 \leq \ldots \leq d_N$; note that, in \lpbound{}, the degrees are in descending order.
When $|X| \leq 1$, we call the degree sequence \emph{simple}, and when $XY = \Attrs(R)$, we say the degree sequence is \emph{full}, and denote it by $\deg_R(\ast | X)$; we will skip the subscript when it is clear from the context.
As constrained by our supported query type (see the above paragraph), we will be using only simple full degree sequences.

\sparagraph{$\ell_p$-Norms.} Instead of storing the degree sequences themselves, we will be storing simple statistics related to them, namely $\ell_p$-norms.
The advantage of operating on norms is that we have to store only a few bytes per relation; in addition, $\ell_p$-norms capture surprisingly well the skew of the degree sequence.
Formally, the $\ell_p$-norm of a sequence $d = (d_1, d_2, \ldots)$ is $\|d\|_p \triangleq (\sum_{i} d_i^p)^{1/p}$, where $p \in (0, \infty]$.
Note that as $p$ goes to $\infty$, the value of the norms decreases and converges to $\max_{i} d_i$.
Finally, notice that the $\ell_p$-norms are invariant to the order of the elements in the underlying sequence.
In addition, as a syntactic sugar, we introduce $\ell_{-\infty} \triangleq \min_{i} d_i$. Usually, this is 1, yet there are relations in the benchmarks where this is larger, e.g., \texttt{movie\_info\_idx}'s \texttt{movie\_id} key-column in the IMDb dataset~\cite{job}.

To denote the $\ell_p$-norm of a given column $X$ in relation $R$, we write $\ell_{R, X, p}$.
As we cannot always compute the norm values exactly, we will work with lower and upper bounds thereof.
To this end, we write $\ell_p^{-}$ and $\ell_{p}^{+}$ to denote its lower and upper bound, respectively.

To understand how prior work on join size upper bounds utilizes $\ell_p$-norms, we refer the reader to related work (see ~\secref{sec:rel-work}).

\subsection{Rearrangement Inequality}\label{sec:re-ineq}

Before we move on to the inner-product reverse inequalities that we will be using,
let us consider an inequality that connects inner products of arbitrarily ordered vectors with
inner products of ordered sequences. This will prove useful for the overall understanding of \xbound{}, and its difference from \lpbound{}~\cite{lpbound}, a state-of-the-art framework for cardinality upper bounds.

\begin{inequality}[\textsc{Rearrangement Inequality}]
  \label{ineq:rearr-ineq}
  Let $a_1 \le \ldots \le a_d$ and $b_1 \le \ldots \le b_d$
  be two non-decreasing sequences of real numbers.
  Then, for every permutation $\sigma$ of $[d]$, it holds that \[
    \sum_{i=1}^d a_ib_{d+1-i}
    \;\le\;
    \sum_{i=1}^d a_ib_{\sigma_i}
    \;\le\;
    \sum_{i=1}^d a_ib_i.
\]
\end{inequality}

In a nutshell, the RHS of the inequality is an \emph{upper bound} on the inner product, while the LHS is a lower bound.
In fact, the RHS is used in \safebound{}~\cite{safebound-1, safebound-2}, a pre-runner of \lpbound{}~\cite{lpbound} (see ~\secref{sec:rel-work} for an exposition).
What we (implicitly) leverage in \xbound{} is the LHS, i.e., taking the inner product between reversely ordered sequences.
Yet, similar to \lpbound{}, we will be directly bounding the inner product by $\ell_p$-norms only.

\subsection{Inner-Product Reverse Inequalities}\label{sec:rev-ineqs}

As motivated in the introduction, we need the so-called reverse inequalities on the inner product of two degree vectors $\va$ and $\vb$.
The caveat is that this class of inequalities requires the input vectors to be strictly positive, i.e., $a_i, b_i > 0, \forall i \in [d]$.
Recall that degree vectors are generally only non-negative, i.e., they contain zeros for non-existing keys in a relation.
In ~\secref{subsec:exact-l0}, we show how we can still apply these inequalities by lower-bounding the number of non-zero entries in both vectors. We start with an intuitive reverse inequality:

\begin{inequality}[\textsc{Min-Degree Inequality}]
  \label{ineq:min-degree}
  Let $\mathbf{a} = (a_1, \ldots, a_d)$ and $\mathbf{b} = (b_1, \ldots, b_d)$ 
  be two positive vectors of real numbers such that
\begin{equation*}
  0 < m_a \le a_i < \infty,\;
  0 < m_b \le b_i < \infty,\;
  \forall i \in [d].
\end{equation*}
  Then
  \[
    \sum_{i=1}^d a_i b_i \geq \max(m_a \cdot \displaystyle\sum_{i=1}^d b_i, m_b \cdot \displaystyle\sum_{i=1}^d a_i).
  \]
\end{inequality}
This inequality tells us that, to lower-bound the inner product, it suffices to assume that one vector is filled with its smallest element. Notably, this requires only the vector's element sum and its smallest element. Next, we continue with a more involved reverse inequality:

\begin{inequality}[\textsc{P\'olya--Szeg\H{o} Inequality}~\cite{ps-ineq}]
  \label{ineq:ps}
  Let $\mathbf{a} = (a_1, \ldots, a_d)$ and $\mathbf{b} = (b_1, \ldots, b_d)$ 
  be two positive vectors of real numbers such that
\begin{equation*}
  0 < m_a \le a_i \le M_a < \infty,\;
  0 < m_b \le b_i \le M_b < \infty,\;
  \forall i \in [d].
\end{equation*}
  Then
  \[
    \sum_{i=1}^d a_i b_i
    \;\ge\;
    2\,
    \frac{
      \sqrt{\displaystyle\sum_{i=1}^d a_i^2}
      \;\sqrt{\displaystyle\sum_{i=1}^d b_i^2}
    }{
      \sqrt{\dfrac{M_a M_b}{m_a m_b}}
      +
      \sqrt{\dfrac{m_a m_b}{M_a M_b}}
    }.
  \]
\end{inequality}
Let us inspect this inequality.
To lower-bound the inner product $\va \cdot \vb$, it requires the following vector norm statistics:
\begin{itemize}
  \item $\ltwo$: $\|\mathbf{a}\|_2$ and $\|\mathbf{b}\|_2$,
  \item $\linfty$: $M_a \triangleq  \max_i a_i$ and $M_b \triangleq  \max_i b_i$,
  \item $\lninfty$: $m_a \triangleq  \min_i a_i$ and $m_b \triangleq \min_i b_i$.
\end{itemize}
As we will estimate the cardinalities of $k$-way joins, we also need the generalization of the above inequality for $k$ vectors.
Fortunately, there is one. Such inequalities form the interesting class of reverse H\"older's inequalities.
As a backbone, we will be using the (currently) tightest inequality in the literature, which indeed generalizes Ineq.~\eqref{ineq:ps}.
Hence, it infers the limitation of relying on $\ltwo$ only. We discuss possible extensions in ~\secref{sec:discussion}.

\begin{figure}
    \centering
    \includegraphics[width=0.75\linewidth]{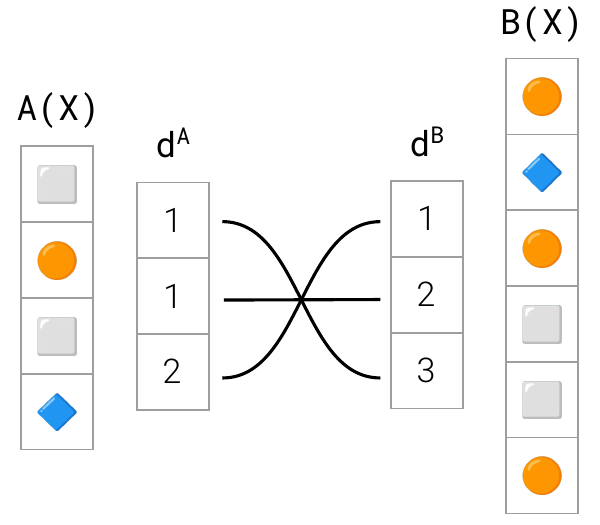}
    \caption{Obtaining a lower bound on the join size, assuming all keys join, by multiplying \emph{cross-wise} the values of the (non-decreasing) degree sequences. This follows from the rearrangement inequality, Ineq.~\eqref{ineq:rearr-ineq}.}
    \label{fig:toy-example}
\end{figure}

\begin{inequality}[\textsc{Generalized Reverse H\"older's Inequality}~\cite{gen-rev-hoelder}]
  Let $\mathbf{v}^{(1)}, \ldots, \mathbf{v}^{(n)}$ be $n \geq 2$ positive vectors of $d$ real numbers each, such that $0 < m_i \leq v^{(i)}_j \leq M_i, i \in [n], j \in [d]$. Then \[
    \prod_{i=1}^{n} \|\mathbf{v}^{(i)}\|_2 \leq \frac{{\sqrt{d}}^{n - 2}}{2^{n - 1}} \left( \prod_{k=2}^{n} B_k\right) \|\mathbf{v}^{(1)} \ldots \mathbf{v}^{(n)}\|_1,
  \]
  where $B_k = \sqrt{\prod_{i=1}^{k} \frac{M_i}{m_i}} + \sqrt{\prod_{i=1}^{k} \frac{m_i}{M_i}}$, for $k \geq 2$.
  \label{ineq:gen-rev-hoelder}
\end{inequality}

First, see that for $n = 2$, the above reduces to the P\'olya--Szeg\H{o} inequality, Ineq.~\eqref{ineq:ps}. Second, note that Ineq.~\eqref{ineq:gen-rev-hoelder} is \emph{permutation-sensitive}, i.e., it does matter in which order the vectors are input.
To see why this is the case, consider the $B_k$ constants: they consider only the statistics of the first $k$ vectors.

\subsubsection{Improved Version.}\label{sec:improved-gen-rev-hoelder}

To obtain a better lower bound in Ineq.~\eqref{ineq:gen-rev-hoelder}, we will consider all (non-redundant) permutations of the input vectors, namely the $(n - 1)!$ permutations obtained after excluding those that differ only by swapping the first two elements, and take the maximum of the resulting lower bounds. While this increases the optimization time, this leads to better bounds.
We briefly mention that it is even possible to have a convolution-like bound, by resorting to a subset dynamic program in time $O(3^n)$; however, we have not observed any improvement in the bounds.
\begin{figure*}
    \centering
    \includegraphics[width=1.0\linewidth]{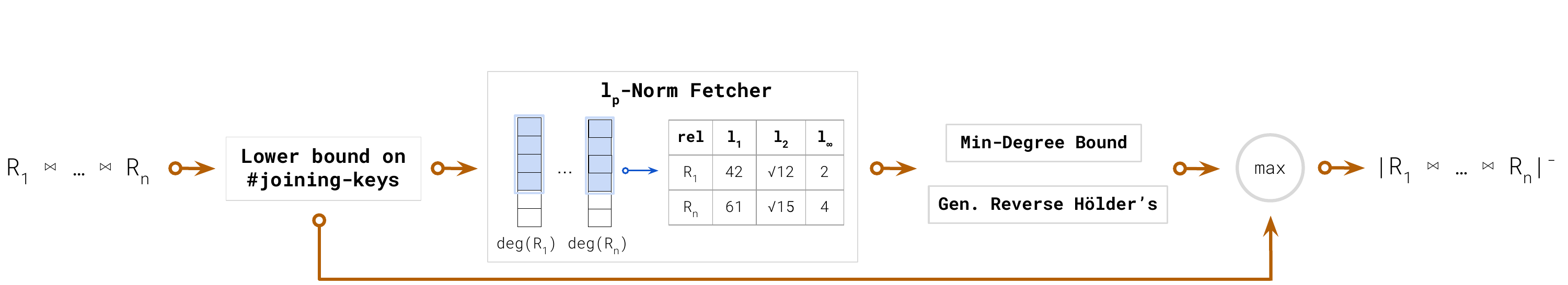}
    \caption{\xbound{}'s estimation method: Given a multi-way join on the same key, it (a) lower-bounds the number of joining keys (= $m$), (b) uses $m$ as the prefix length for evaluating degree-sequence norms on the base tables, and (c) takes the best bound from the generalized reverse H\"older's inequality, \texttt{min-degree bound}, and $m$. The output represents a lower bound on the query size. Note that this visualization does not cover the heavy partition (\secref{subsec:heavy-partition}).}
    \label{fig:overview}
\end{figure*}

\section{Join Size Lower Bounds}\label{sec:xbound}

\sparagraph{Intuition.} We provide the key idea behind join size lower bounds. Consider the toy example in Fig.~\ref{fig:toy-example}, where we show two join columns, $A(X)$ and $B(X)$, with their corresponding degree sequences on the side. For the sake of presentation, we assume that each key has a join partner in the other table. The join cardinality can be computed as: 2 $\times$ 2 ($\Box$) + 1 $\times$ 3 ($\Circle$) + 1 $\times$ 1 (\rotatebox[origin=c]{45}{$\Box$}) = 8.

The degree sequences depicted on the side do not capture the keys for which the frequencies are attained.
The question is thus: What is a valid (non-zero) lower bound in this case?
A moment of thought reveals that multiplying \emph{cross-wise} the values in the degree sequences yields the best we can hope for.\footnote{See the visualization; this also inspired our estimator's acronym.} In this case, the lower bound is 1 $\times$ 3 + 1 $\times$ 2 + 2 $\times$ 1 = 7. This does have a mathematical argument behind it, namely the \emph{rearrangement inequality}; we already discussed it in ~\secref{sec:re-ineq}.

On the other side, \safebound{} (and, extrapolated, \lpbound{}) uses element-wise multiplication to upper bound the join size. In this example, the upper bound reads: 1 $\times$ 1 + 1 $\times$ 2 + 2 $\times$ 3 = 9.

\sparagraph{From Degree Sequences to Norms.} Surely, we could operate on degree sequences, similar to what \safebound~\cite{safebound-1, safebound-2} did for upper bounds in the context of Berge-acyclic queries.
However, simply storing them might yield increased storage costs.
Instead, we will follow the approach of \lpbound{}, namely operating with the associated norms of the degree sequences.
The main disadvantage is that $\ell_p$-norms are agnostic to the order of the elements in the vector they are applied to.

\sparagraph{Overview.} We visualize \xbound{}'s estimation strategy in Fig.~\ref{fig:overview}: First, (a) we infer a lower bound on the number of \emph{distinct} joining keys. This is essentially the length of the degree sequences which we should take the norms from. Then, (b) we estimate the norms needed for the reverse inequalities in ~\secref{sec:rev-ineqs}, such as $\ell_1$, $\ell_{2}$, and $\ell_{\infty}$.
By estimating, we mean computing a strict lower and upper bound of them. Finally, (c) we lower-bound the inner product by taking the maximum of all bounds output by the set of inequalities. Note that the value from (a) is also a valid lower bound.

Next, we start with the very first step: obtaining hard lower bounds the number of distinct joining keys.

\subsection{Hard $\ell_0$ Lower Bounds}\label{subsec:exact-l0}

\sparagraph{Why This is Needed.} The attentive reader might have remarked that our toy example in Fig.~\ref{fig:toy-example} does not entirely reflect the reality.
Indeed, it is not always the case that the two relations will have the same set of keys.
As a consequence, if we considered the full key domain, our degree sequences would have been padded with zeros.
This is detrimental to the application of the inequalities we presented in ~\secref{sec:rev-ineqs}, as they all require \emph{positive} vector entries.

Our key insight is that we can first infer a lower bound on the number of joining keys, essentially $\lb{\ell_{A\,\bowtie\,B,\,X,\,0}}$, and then apply the inequalities with norms confined to the first $\lb{\ell_{A\,\bowtie\,B,\,X,\,0}}$ entries in the degree sequences of both relations. Note that this alone is a valid lower bound on the join size. As we will later show, inner-product inequalities can boost it even further.

\sparagraph{Two-Way Joins.} Let $\mathcal{K}_A$ and $\mathcal{K}_B$ be the set of keys in $A(X)$ and $B(X)$, respectively.
Bounding the number of joining keys is equivalent to bounding $|\K_{A} \cap \K_{B}|$.
Surely, having at our disposal only the cardinalities of the two sets, $\K_A$ and $\K_A$, we \emph{cannot} have a better lower bound than 0.
The key realization is that we can use the min/max-values of the columns, commonly referred to as zonemaps~\cite{zonemaps}, which are a standard feature in modern file formats~\cite{parquet, btr-blocks, vortex, fast-lanes}.
Let us understand why this changes the whole picture.
Note that we can re-write $|\K_A \cap \K_B|$ as follows:
\begin{equation}
  |\K_A \cap \K_B| = |\K_A| + |\K_B| - |\K_A \cup \K_B|.
  \label{eq:eq-2-rel-base}
\end{equation}

Thus, to lower-bound $|\K_A \cap \K_B|$ we are to \emph{upper-bound} $|\K_A \cup \K_B|$ (due to the negative sign).
In turn, the latter can be bounded from above by the span size, i.e., $\max\{ \K_A, \K_B \} - \min \{ \K_A, \K_B \}$.
With this we have:
\begin{equation}
  |\K_A \cap \K_B| \geq \max \{ 0, |\K_A| + |\K_B| - (\max\{ \K_A, \K_B \} - \min \{ \K_A, \K_B \})\}.
  \label{eq:eq-2-rel}
\end{equation}

\sparagraph{Example.} Let us consider an example. Take as domain the set $\{1, \ldots, 100\}$, and let the two sets have cardinalities $|\K_A|$ = 75 and $|\K_B|$ = 50.
What is, in this case, a valid lower bound on the intersection size? At a closer inspection, the two sets can be at the two ends of the range, and thus they intersect only on 25 elements.
This exactly matches Eq.~\eqref{eq:eq-2-rel}: $|\K_A \cap \K_B| \geq 75 + 50 - 100 = 25$.

We now outline a simple generalization to multi-way joins.

\sparagraph{Multi-Way Joins.} Note that we can generalize Eq.~\eqref{eq:eq-2-rel-base} as follows. Given $n$ sets $\{\K_i\}_{i=1}^{n}$, we are to lower-bound their intersection size. Note again that we do not dispose of the sets themselves,
but of limited statistics related to them, such as (bounds on) their cardinalities and their domain size.

To this end, we will be employing a lightweight generalization of the two-way join case. Denote the universe as $U \vcentcolon= \bigcup_{i}^{n} \K_i$.
Then:
\begin{align}
  \bigl|\;\bigcap_{i=1}^{n} \K_i\;\bigr|
    &= |U| - \Bigl|\;\bigcup_{i=1}^{n}\;(U \setminus \K_i)\;\Bigr| \notag\\
    &\ge |U| - \sum_{i=1}^{n} |U \setminus \K_i| \notag\\
    &= |U| - \sum_{i=1}^{n} \bigl(|U| - |\K_i|\bigr) \notag\\
    &= |U| - n|U| + \sum_{i=1}^{n} |\K_i| \notag\\
    &= \sum_{i=1}^{n} |\K_i| - (n - 1)|U|,
    \label{eq:gen-ndv-lb}
\end{align}
where we used the union bound in the second step. Note that, for the case $n = 2$, we obtain exactly the two-way join case, Eq.~\eqref{eq:eq-2-rel-base}.

\subsection{Probabilistic $\ell_0$ Lower Bounds}\label{subsec:prob-l0}

The above section discussed how to get \emph{hard} lower bounds on the number of distinct joining keys. However, we noticed that collecting the \emph{exact} $\ell_0$ statistics can be slow, especially for larger tables. Several database systems, including Fabric DW, already store on their catalogs HyperLogLog~\cite{hll} sketches of their columns. Indeed, these can be used to get a \emph{probabilistic} lower bound on $\ell_0$, by plugging them in Eq.~\eqref{eq:gen-ndv-lb} for the $|\K_i|$ values. Alternatively, one can use the inclusion-exclusion principle, yet this leads to high-variant estimates~\cite[Sec.~2.2]{theta-sketch}

In fact, there are specialized sketches that can estimate the intersection size of $n$ sets directly~\cite{theta-sketch, set-sketch}. To this end, we used the \texttt{ThetaSketch}~\cite{theta-sketch}, as it is readily provided in Apache's DataSketches library. Hence, instead of lower bounding $|\cap_{i=1}^{n} \K_i|$ via the union bound---recall Eq.~\eqref{eq:gen-ndv-lb}, we can get a better lower bound, with a 99\% confidence, by intersecting the Theta sketches corresponding to the $n$ sets. These sketches are stored as pre-computed base table statistics; we come later in \secref{sec:base-stats} onto how we manage them.

Notably, since the inner-product inequalities are deterministic functions of the estimated intersection lower bound, the 99\% confidence guarantee provided by the Theta sketch propagates unchanged to the derived query size lower bounds.

\sparagraph{Outlook.} This concludes our approach to lower-bounding the number of distinct joining keys for an arbitrary number of joins on the same join key.
In a subsequent section (\secref{subsec:light-partitions}), we will further develop it, by partitioning the value range into equi-width bins, and applying the above strategy within each bin.
Until then, we address the aspect of how to leverage this very lower bound to enable reverse inner-product inequalities, which will further boost it.

\subsection{Enabling Reverse Inequalities}

We show how the last section leads to a valid application of the reverse inequalities in ~\secref{sec:rev-ineqs}, which require that the vectors have positive entries only.

\sparagraph{Intuition.} Let us start with the simplest case, the two-way join.
Let $A(X, \ldots)$ and $B(X, \ldots)$ be two relations which join on $X$, and let $D$ be the support of both $A(X)$ and $B(X)$. 
By the above considerations (\secref{subsec:exact-l0}), we know the lower bound $\ell_{A\,\bowtie\,B,\,X,\,0}^- \vcentcolon= m$ on the number of distinct $X$ keys of the join $A \bowtie B$.
Note that this tells us a little bit more. Indeed, let $\va$ and $\vb$ be the degree vectors of column $X$ in $A$ and $B$, respectively.
Recall that the lower bound $m$ simply tells us that we can find a set $\K \subseteq A(X) \cap B(X)$ of size $|K| = m$, such that $\va[k] \cdot \vb[k] \neq 0, \forall k \in \K$.
This, in turn, implies that both $\va$ and $\vb$ take \emph{positive} values on $\K$, which makes them amenable to inner-product reverse inequalities. 

Moreover, the set $\K$ is only a subset of the \emph{actual} set of joining keys.
In other words, the sum $\sum_{k \in K} \va[k] \cdot \vb[k]$ is a lower bound on the actual join size, $\sum_{k \in D} \va[k] \cdot \vb[k]$.

\sparagraph{Two-Way Joins.} The above explanation has a caveat: we actually do not have access to the degree vectors themselves. Surely, we could store them, as argued in the beginning of ~\secref{sec:xbound}, but this is space-inefficient.
What we do have access to are the norms of the (corresponding) degree \emph{sequences}. With these, we can lower-bound the aforementioned sum $\sum_{k \in K} \va[k] \cdot \vb[k]$ by exploiting the \emph{prefix} of length $m$ in both degree sequences, as follows:

\begin{lemma}[Prefixes of Degree Sequences; two-way join]
  \label{lemma:cw-ds}
  Let $\va$ and $\vb$ be the degree vectors of column $X$ in relations $A$ and $B$, respectively, and let $\K$ denote a subset of the joining keys, i.e., $\va[k] \cdot \vb[k] \neq 0, \forall k \in \K$.
  Moreover, let $d^A$ and $d^B$ be the corresponding degree sequences.
  Then, there exists a permutation $\sigma : [|\K|] \to [|\K|]$, such that \[
    \sum_{k \in \K} \va[k] \vb[k] \geq \sum_{i=1}^{|\K|} d^A_i d^B_{\sigma_i}.  
  \]
\end{lemma}
\begin{proof}
  Let $(\hat{a}_i)_{i=1}^{|\K|}$ and $(\hat{b}_i)_{i=1}^{|\K|}$ denote the (non-decreasing) degree sequences of $X$ in $A$ and $B$, respectively,
  yet with values from $\va[\K]$ and $\vb[\K]$ only.
  Intuitively, these are \emph{subsequences} of the actual degree sequences $d^A$ and $d^B$.
  
  Observe now that there is a permutation $\sigma : [|\K|] \to [|\K|]$, such that we can re-write the initial inner product as \[
    \sum_{k \in \K} \va[k] \vb[k] \triangleq \sum_{i=1}^{|K|} \hat{a}_i \hat{b}_{\sigma_i}.
  \]
  In turn, it also holds that $\hat{a}_i \geq d^A_i$, $\forall i \in [|\K|]$; analogously for $\hat{b}$. With this, we obtain \[
    \sum_{i=1}^{|K|} \hat{a}_i \hat{b}_{\sigma_i} \geq \sum_{i=1}^{|K|} d^A_i d^B_{\sigma_i},
  \]
  which concludes the proof.
\end{proof}

This lemma allows us to connect a lower bound on the number of distinct joining keys with a inner product on the degree sequences.
And this is a key observation: We can now apply \emph{any} inner-product reverse inequality, as degree sequences have positive entries only.

\sparagraph{Example 1.} Let us consider a concrete example. We take the simplest reverse inequality, Ineq.~\eqref{ineq:min-degree}, which requires only the $\ell_1$ and the $\ell_{-\infty}$ statistics. Assume we already computed a lower bound $m = 4$ on the number of distinct joining keys. Consider the two degree sequences: \[
  d^A = (\colA{1}, \colA{1}, \colA{1}, \colA{2}, \colA{2}, \colA{3}),\qquad
  d^B = (\colB{1}, \colB{1}, \colB{2}, \colB{3}, \colB{3}, \colB{4}).
\]
The lower bound using the degree sequences reads: \colA{1} $\times$ \colB{3} + \colA{1} $\times$ \colB{2} + \colA{1} $\times$ \colB{1} + \colA{2} $\times$ \colB{1} = 8, as we only consider the prefix of length $m = 4$. Recall, however, that we operate only with the statistics of these degree sequences. Calculating them from the (prefixed) degree sequences, we obtain \[
  \begin{array}{c|c|c}
    \text{degree sequence} & \ell_1 & \ell_{-\infty} \\
    \hline
    d^A[:4] & \colA{1} + \colA{1} + \colA{1} + \colA{2} = \colA{5} & \colA{1}\\
    d^B[:4] & \colB{1} + \colB{1} + \colB{2} + \colB{3} = \colA{7} & \colB{1}
  \end{array}
\]
Therefore, this tells us that (i) \colA{5} keys from $A$ can join with \colB{1} key from $B$ and (ii) \colB{7} keys from $B$ can join with \colA{1} key from $A$. Taking the maximum of these two options yields a lower bound 7. Note that this is slightly worse than the lower bound via the degree sequences directly.

This is what we refer to as the \texttt{min-degree} bound, and is formalized in the upcoming \secref{subsec:min-degree}. Next, we consider another example, where we consider an $\ell_2$-based reverse inequality. 

\sparagraph{Example 2.} We would like to apply the P\'olya--Szeg\H{o} inequality, Ineq.~\eqref{ineq:ps}, and we know a lower bound $m = 4$ on the number of distinct joining keys, as in the previous example. Moreover, the two degree sequences stay the same.

To apply the P\'olya--Szeg\H{o} inequality, Ineq.~\eqref{ineq:ps}, we need the values of $\ell_2, \ell_{\infty}$, and $\ell_{-\infty}$. Calculating them from the (prefixed) degree sequences, we obtain \[
  \begin{array}{c|c|c|c}
    \text{degree sequence} & \ell_2 & \ell_\infty & \ell_{-\infty} \\
    \hline
    d^A[:4] & \sqrt{\colA{1}^2 + \colA{1}^2 + \colA{1}^2 + \colA{2}^2} = \colA{\sqrt{7}} & \colA{2} & \colA{1}\\
    d^B[:4] & \sqrt{\colB{1}^2 + \colB{1}^2 + \colB{2}^2 + \colB{3}^2} = \colB{\sqrt{15}} & \colB{3} & \colB{1}
  \end{array}
\]
Therefore, we can lower-bound the join cardinality as \[
  |A \bowtie B| \geq 2 \frac{\colA{\sqrt{7}} \cdot \colB{\sqrt{15}}}{
     \sqrt{\frac{\colA{2} \cdot \colB{3}}{\colA{1} \cdot \colB{1}}}
        +
    \sqrt{\frac{\colA{1} \cdot \colB{1}}{\colA{2} \cdot \colB{3}}}
  } \geq 7.17.
\]
This is (slightly) better than the lower bound achieved by the \texttt{min-degree} bound. This concludes the examples for two-way joins.

The topic of actually providing the base table norms will be the topic of ~\secref{sec:base-stats}.
So far, we simply described how to enable inner-product reverse inequalities, by exploiting the lower bound on the number of distinct joining keys calculated in ~\secref{subsec:exact-l0}.
Next, we will generalize the above mechanism to multi-way joins.

\sparagraph{Multi-Way Joins.} In the sequel, we will be considering a join over the same attribute $X$, i.e., we have $n$ relations $A^{(i)}(X, \ldots)$ and $\va^{(i)}$ the corresponding degree vectors, $i \in [n]$.
In addition, denote the domain $D$ as the support of all $A^{(i)}(X)$.
Similar to the two-way setting, we will have to first lower bound the number of distinct joining keys that the join produces. This has already been treated in ~\secref{subsec:exact-l0}.
We next show how to improve on this value by using reverse inner-product inequalities.

We will just need a generalization of the two-way case.
Let $m$ be the lower bound on its cardinality as in ~\secref{subsec:exact-l0}, i.e., \[
  m \vcentcolon= \ell_{\bowtie_{i=1}^{n} A^{(i)},\:X,\:0}^-.
\]
Recall that this means that we can find a subset $\K$ of the joining keys set, i.e., $\K \subseteq \cap_{i=1}^{n} A^{(i)}$.
In turn, the sum $\sum_{k \in \K} \prod_{i=1}^{n} \va^{(i)}[k]$ is indeed a lower bound on the actual join size, $\sum_{k \in D} \prod_{i=1}^{n} \va^{(i)}[k]$, since $\K \subseteq D$.
To come again to the realm of degree sequence norms, we need a similar result to Lemma~\ref{lemma:cw-ds}:

\begin{lemma}[Prefixes of Degree Sequences; multi-way join]
  Let $\va^{(i)}$ and $d^{(i)}$ be the degree vectors and sequences, respectively, of column $X$ in relations $A^{(i)}$, $\forall i \in [n]$, and $\K$ a subset of the joining keys, i.e., $\prod_{i=1}^{n} \va^{(i)}[k] \neq 0, \forall k \in \K$.
  Then, there exist permutations $\sigma^{(i)} : [|\K|] \to [|\K|]$, $\forall i \in [2, n]$, such that \[
    \sum_{k \in \K} \prod_{i=1}^{n} \va^{(i)}[k] \geq \sum_{j=1}^{|\K|} d^{(1)}_{j} \prod_{i=2}^{n} d^{(i)}_{\sigma^{(i)}_j}.
  \]
  \label{lemma:many-cw-ds}
\end{lemma}
\begin{proof}
  The proof is a simple generalization of that of Lemma~\ref{lemma:cw-ds}.
\end{proof}

The result tells us that we can move away from degree vectors and instead work solely with the prefixes (of length $m$) of the corresponding degree sequences. This enables the application of all reverse inner-product inequalities described in the preliminaries, in particular the generalized reverse H\"older's inequality, Ineq.~\eqref{ineq:gen-rev-hoelder}.

\subsection{Min-Degree Bound}\label{subsec:min-degree}

As pointed out in our first example, we can use Ineq.~\eqref{ineq:min-degree} to have what we term the \texttt{min-degree} bound, which is reminiscent of the \texttt{max-degree} bound in \panda{}~\cite{panda}. Compared to the other reverse inequalities, it does not require the $\ell_2$ norm, which is not available by default on database catalogs.

The key idea is that, once we have the lower bound on the number of distinct joining keys, $m$, we can fix a relation $R_i$, and assume that all other relations join with $\ell_{-\infty}$ keys each. Formally, this reads as follows: \[
  \displaystyle\max_{i} \ell_{R_i, X, 1} \displaystyle\prod_{j \neq i} \ell_{R_j, X, \ell_{-\infty}}.
\]
Note that norm values are taken, as usual, from the prefix of length $m$ of the corresponding degree sequences.

\subsection{Heavy Partition}\label{subsec:heavy-partition}

For a two-way join, frameworks for upper bounds on query sizes, such as \lpbound{} and \safebound{}, tend to ignore which keys join from the two tables. And, indeed, this is a pragmatic approach, since keeping track of the joining keys themselves is expensive. In our experiments, however, we saw that lower bounds cannot be competitive enough if we do not maintain, at least, a small set of heavy hitters that are going to be joined. Hence, apart from the (light) partitions we will be introducing later (\secref{subsec:light-partitions}), we will maintain on each table a set of heavy hitters, which we refer to as the \emph{heavy partition}. For a given join pool (recall \secref{sec:prelims}), we maintain a set $\mathcal{H}$ that represents the heavy hitters across the corresponding join columns. We keep this set constant across all supported predicate types, and limit its cardinality to $\mathcal{H}$ to a small constant (set in our experiments to 64). To this end, we used the \texttt{FrequentItems} sketch~\cite{frequent-items} available in Apache's DataSketches library, with a precision of 12, and extracted the heavy keys without false positives, along with a (hard) lower bound on their degrees; this guarantees a single-pass computation.

At estimation time, the heavy partition is tracked of separately, namely: we track which heavy keys survive the join, and then collect their exact degrees (or lower bounds thereof) from the pre-computed statistics. We discuss how to compute these statistics in \secref{subsec:pred-and-heavy}. For non-selective queries, integrating a heavy partition into \xbound{} enabled almost perfect estimates (only some constant factors away), while preserving the guarantee of hard lower bounds.

Note that Fabric DW already provides a way to keep track of a column's heavy hitters and their (approximate) degrees, by employing a \texttt{CountMinSketch}~\cite{count-min-sketch}. Hence, \xbound{}'s heavy partition can be supported without much extra overhead.
\section{Base Table Statistics}\label{sec:base-stats}

The previous section enabled the application of reverse inner-product inequalities to the setting of degree sequences, and implicitly the norms associated with them.
The remaining open question is how to provide these norms to the inequalities.
On a high level, given the prefix length $m$, we need to extract the required norms from the degree sequences.
In database terminology, we need prefix norms with random access.

\sparagraph{Prefixes.} The easiest solution is to maintain the norms on prefixes of length a power of two. Indeed, this is resemblant to the prefix trick in \lpbound~\cite[Sec.~5; Optimizations]{lpbound}.
\emph{However}, note that this optimization in the context of \lpbound{} is simply a booster for the quality of base table norms it returns;
indeed, one could always return the norms of the whole table.
For \xbound{}, \emph{unfortunately}, this is an indispensable component to be able to provide any lower bounds.

An attentive reader would note that simply using the \emph{prefix trick} does not actually lead to the \emph{exact} value of the wished-for norms.
Surprisingly, we do \emph{not} need the actual values.
Note that we can further lower-bound the RHS of any reverse inner-product inequality by providing hard lower and upper bounds on the requested norms. For instance, to lower bound the RHS side of the P\'olya--Szeg\H{o} inequality,
we need (i) lower bounds on $\ell_2$, $\ell_{\infty}$, and $\ell_{-\infty}$, and (ii) upper bounds on $\ell_{\infty}$ and $\ell_{-\infty}$.
We could infer these lower bounds on the norms by taking a look at the immediate prefixes sizes around $m$, the provided prefixed length.
Indeed, for $m$ = 6, we could take as a lower bound for $\ell_2$ the value of length 4,
while for an upper bound of $\ell_{\infty}$, we could take the value of the prefix length 8, as the degree sequence is sorted in non-decreasing order.
While this suffices for correctness, i.e., these are hard lower and upper bounds on the requested norms, there is a way how to further boost their quality.
We call this trick \emph{norm stitching}.

\subsection{Norm Stitching}\label{subsec:norm-stitching}

\begin{figure}
    \centering
    \includegraphics[width=0.5\linewidth]{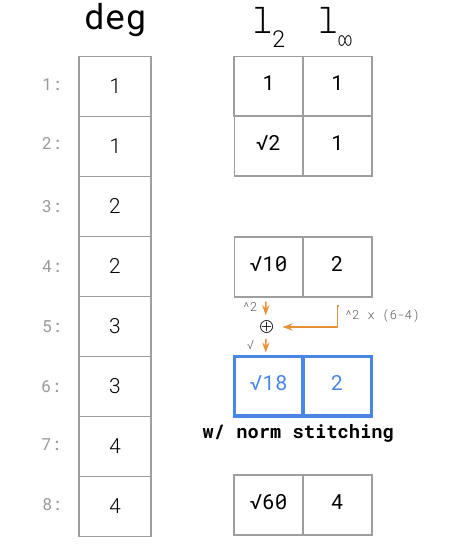}
    \caption{Norm stitching: Deriving $\ell_2$ lower bounds for non-power-of-two prefixes by extending the previous power-of-two prefix using its $\ell_\infty$ value.}
    \label{fig:norm-stitching}
\end{figure}

The idea is to ``stitch'' the norm values at the prefixes of length power of two.
Consider Fig.~\ref{fig:norm-stitching}, which visualizes the idea behind norm stitching. We would like to extract the norm values for $m = 6$.
Yet, we only have the values at prefixes 4 and 8. To this end, to obtain a lower bound on $\ell_2$, as requested by the P\'olya--Szeg\H{o} inequality,
we can take the $\ell_2$-value at 4, and artificially add (constant) degree values that are guaranteed to be below the actual ones.
In this case, we add the degree value $2$ twice, since that is the maximum degree up until prefix length $4$,
and we know that the following degree values will be \emph{at least} that.
We thus obtain: $(\sqrt{10})^2 + 2^2 \times (6 - 4) = 18$. By taking the root, we obtain a lower bound on the actual $\ell_2$ value at prefix 6, which is $\sqrt{28}$.

In pseudo-code, norm stitching for $\ell_2$ works as follows:
~\\
\begin{lstlisting}[label=l:norm-stitching,language=Python]
 def NormStitching(norms, rel_name, m):
   deg_seq_len = 2**i s.t. 2**i <= m
   l_2 = norms[rel_name][deg_seq_len]['l_2']
   l_infty = norms[rel_name][deg_seq_len]['l_+infty']
   return math.sqrt(
     l_2**2 + (m - deg_seq_len) * l_infty**2
   )
\end{lstlisting}

Namely, given the prefix length $m$, we fetch the $\ell_2$ value for the \emph{largest} stored prefix (in our case, only power-of-two prefix lengths) as well as its corresponding $\ell_\infty$ value. We then stitch the $\ell_2$ norm by adding as many $\ell_\infty$ contributions as needed to account for the gap.

\sparagraph{Downwards Stitching.} Notably, it is also possible to do a ``downwards'' norm stitching: Instead of padding the degree sequence from the preceding power of two, one can pad it from the subsequent power of two. Namely, to obtain a lower bound on $\ell_2$ at prefix 6, one can deduct as many $\ell_\infty$'s as the gap. In this particular example, we would obtain: $\sqrt{60}^2 - 4^2 \times (8 - 6)$ = 28, and by taking the root we obtain, in this particular case, the actual value of the norm. \xbound{} implements both the upwards and downwards stitching strategies, returning the larger of the two resulting bounds.

We now come to the last component of \xbound{} that enhances the quality of the lower bounds on the number of joining keys.

\subsection{Light Partitions}\label{subsec:light-partitions}

Lower bounds on the number of joining keys based only on simple statistics, such as the $\ell_0$-value of each column, and the respective min/max values,
can be rather loose. This holds also for the extension to probabilistic $\ell_0$ in lower bounds outlined in \secref{subsec:prob-l0}. To this end, we introduced in \secref{subsec:heavy-partition} the idea of a heavy partition, which is targeted at heavy hitters. For the rest of the keys, one solution is to zoom in on the range value to understand where the joining keys actually occur. Hence, we introduce a lightweight partitioning trick that first splits the value range of each column into a fixed number $B$ of equi-width bins,
and collects the same statistics as in ~\secref{subsec:exact-l0}. We refer to these as the \emph{light} partitions. Note that the partitions for which we collect min/max values can be independent of the actual physical table partitions. To exemplify this, note that the light partitions can also be done a via an arbitrary hash function when using the probabilistic lower bounds; the hard lower bounds do need the min/max statistics, as explained \secref{subsec:exact-l0}.

Let us outline the benefits this approach.
First, collecting $\ell_0$ statistics within the $B$ ranges is lightweight (for small $B$'s). Second, this provides us with more fine-grained lower bounds on the number of joining keys, leading to better lower and upper bounds on the norm values that are plugged into the inner-product inequalities, especially when using norm stitching (\secref{subsec:norm-stitching}).

\subsection{Predicate Types}

The above considerations hold for non-filtered base tables. We next show how to still provide lower bounds for the most common predicate types.

\subsubsection{Equality Predicates.}\label{subsubsec:equality-preds} Similar to \lpbound{}, we consider the most common values (MCVs) of a given predicate column, and build the exact same norms as discussed previously.
Concretely, for an MCV $a$ on column $A$ of relation $R$, we compute the set of $\ell_p$ norms needed in the reverse inequalities on top of the degree sequence $\text{deg}_R(\ast\:|\:X_i, A = a)$, for all join columns $X_i$ of relation $R$.

At query time, for a predicate $A = v$, if $v$ is an MCV, then we simply take the corresponding norms (boosted with norm stitching, eventually). However, if $v$ is \emph{not} an MCV, we simply return 0-valued norms.

\subsubsection{Range Predicates.}\label{subsubsec:range-preds} We borrow the hierarchical histogram of \lpbound{}, yet adapt it for lower bounds, namely: Each layer is a histogram whose number of buckets is half the number of buckets of the histogram at the layer below and covers the entire domain range of the predicate column $A$. For a histogram bucket with boundaries [$x_i$, $y_i$], we compute the $\ell_p$-norms on $\text{deg}_R(\ast\:|\:X_i, A \in [x_i, y_i])$. 

At runtime, to fetch the bounds for a range predicate $A \in [x, y]$, we select (i) for the lower bound, the largest bucket that \emph{is contained} in $[x, y]$, and (ii) for the upper bound, the smallest bucket that \emph{contains} $[x, y]$.  However, we observed that, on large datasets, supporting the prefixed norms in this case adds a considerable overhead when pre-computing the statistics. This is due to the fact that we did not exploit the overlaps naturally found in the range buckets. As a tradeoff, we chose the take at estimation time, as a valid lower bound, solely the $\ell_0$ lower bound on the corresponding bucket, without relying on the $\ell_p$ norms.

\subsubsection{Conjunctions \& Disjunctions.}\label{subsubsec:conj-and-disj}

Let us start with disjunctions, i.e., \texttt{pred$_1$} \textcolor{blue}{\texttt{or}} \texttt{...} \textcolor{blue}{\texttt{or}} \texttt{pred$_p$}, since they are more intuitive. In \lpbound{}, these are estimated by taking, due to Minkowski's inequality, the sum of the individual norms. Notably, in \xbound{}, since at least one of the disjuncts should qualify, we can simply take the maximum $\ell_0$ in the first step of lower-bounding the number of joining keys, $m$. Once $m$ has been reported, we retrieve the norm vectors for all predicates and compute their coordinate-wise maximum; that is, for each norm we take the largest value among the predicates.

Conjunctions, i.e., \texttt{pred$_1$} \textcolor{blue}{\texttt{and}} \texttt{...} \textcolor{blue}{\texttt{and}} \texttt{pred$_p$}, are slightly more delicate, as they can more easily yield a zero lower bound. Since we have to first lower-bound the number of joining keys, $m$, we can treat the conjuncts as individual predicates, and report their $\ell_0$ to the first step of lower-bounding the number of joining keys, $m$. Once we have determined $m$, we will do the same strategy as for disjunctions, as it is guaranteed that, for the given $m$ keys, all individual conjuncts evaluate to true.

We continue with the discussion on the support for other predicate types in \secref{sec:discussion}. We next show how to support the above predicates when partitioning (\secref{subsec:light-partitions}) is enabled.

\sparagraph{Predicates $\times$ Partitioning.} To integrate with the partitioning trick, one has to extend the support for the above predicates, by performing the very same trick as described in \secref{subsec:light-partitions} for each MCV and histogram bucket, respectively. We stress that fact that we store per partition solely the statistics related to the distinct key count (or the Theta sketches, in the case of the probabilistic $\ell_0$ lower bounds) and the min/max values of the key.

Note that we do \emph{not} store any $\ell_p$-norms. It is indeed possible to extend \xbound{} to do so---with presumably better lower bounds, yet this adds another layer of norms, making the computation of statistics computationally intensive. Since 

\subsection{Predicates $\times$ Heavy Keys}\label{subsec:pred-and-heavy}

As a booster for our lower bounds, we introduced in \secref{subsec:heavy-partition} the concept of a heavy partition: heavy hitters appearing on the join columns. In the following, we outline how they are supported in combination with the predicates. The statistics are stored the \texttt{hh\_stats} table. First, in the case of no predicates, we store per heavy join key its degree (frequency) on each table. Equality predicates are supported via the MCV support: For each MCV on the predicate column, we store in \texttt{hh\_stats} the degrees of the heavy keys; similarly for the range predicates, where we store the degree w.r.t.~to a range bucket.

\sparagraph{Multiple Predicates.} While this yields exact lower bounds when having table scans with at most one predicate, the case of multiple predicates incurs a challenge. Namely, we have to lower bound the degree of a heavy key under two or more predicates, which, without any additional information, can only be set to 0. While this can be well pre-computed for two predicates, e.g., storing the degrees under an MCV and a range bucket, this can be quite of an overhead. Instead, we opted for a probabilistic approach: per predicate type, we can store a \texttt{ThetaSketch} on the tuple-id in \texttt{hh\_info}, and, at estimation time, lower-bound, with a given confidence, the intersection size of the sets represented by the sketches. Indeed, this gives a probabilistic lower bound on the number of tuples that match all predicates, i.e., on the degree of the respective heavy key. Notably, applying this trick for all heavy keys weakens the guarantees of the lower bounds. We plan to investigate the resulting confidence of these probabilistic lower bounds.
\begin{figure*}
    \centering
    \includegraphics[width=1.0\linewidth]{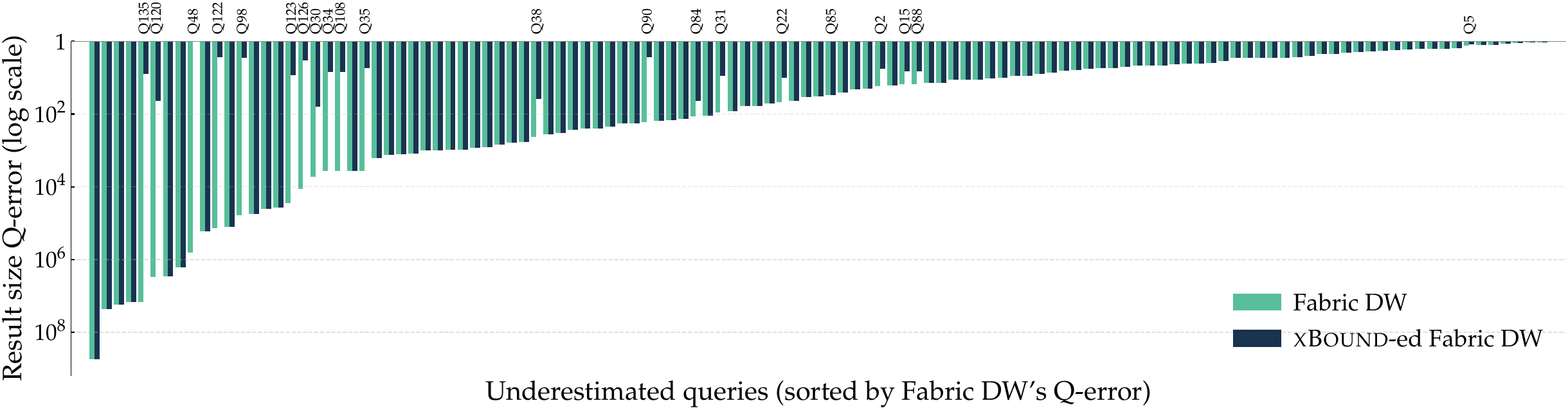}
    \caption{Reducing Fabric DW's query underestimates via clipping with \xbound{}'s lower bounds on the SO-CEB benchmark. We show the queries in order of their Fabric DW's Q-error and highlight the queries with improvements.}
    \label{fig:card-q-error}
\end{figure*}

\section{Experiments}\label{sec:experiments}

We now empirically evaluate the effectiveness of \xbound{} in reducing cardinality underestimation in a real production system, Fabric DW. Particularly, we demonstrate that clipping Fabric DW's estimates with \xbound{}'s lower bounds directly improves resource estimation, yielding sizable query speedups.

\sparagraph{Benchmark.} Since Fabric DW is a cloud-based distributed data warehouse, we need a dataset large enough to exercise its full capability. We therefore run the \textsc{STATS-CEB} benchmark~\cite{stats-ceb} on the 220 GB StackOverflow dump~\cite{sqlstorm}, ensuring distributed operators appear in our system's query plans; we refer to this as \textsc{StackOverflow-CEB} (\textsc{SO-CEB}). Note that we remove filters on the \texttt{posts.favoritecount} column, as it is \texttt{NULL} in this dump. The benchmark comprises 146 Berge-acyclic queries over up to 7 tables with equality and range predicates on numeric and date columns.\footnote{We consider their simplest form: \texttt{select count(*) from [...].}} Of these, \xbound{} currently supports 89~/~146 queries whose (mostly FK-FK) joins share a single join key; specifically, they join either on \texttt{userid} or \texttt{postid}, including semantically equivalent attributes such as \texttt{posts.owneruserid} and \texttt{postlinks.relatedpostid}.

\sparagraph{System Configuration.} The experiments on Fabric DW are conducted on Microsoft Fabric with a F64 capacity. The standalone statistics required by \xbound{} are built on a single node equipped with an Intel® Xeon® Gold 5318Y CPU (24 cores, 48 hyper-threads) processor and 128 GB DDR4 main memory, running Ubuntu 24.04. To this end, we use DuckDB v1.4.3.

Unless otherwise specified, \xbound{} uses the following default configuration: 64 heavy keys, 16 light partitions, 1024 MCVs for equality predicates, and 128 histogram buckets for range predicates. We activate the following optimizations: norm stitching (\secref{subsec:norm-stitching}), heavy keys (\secref{subsec:heavy-partition}), which includes the extension to multiple predicates (\secref{subsec:pred-and-heavy}).

\subsection{Reducing Underestimation}

Fabric DW's resource estimator uses the input data volume as a proxy for estimating how many CPUs are required by the query. The data volume of an operator is computed by the number of rows times the average row size. To this end, we observed that the last join operator in the plan accounted for most of the underestimation. Since the row size, in the case of \texttt{select count(*)} queries, is just a constant, we directly report the estimation error on the number of rows.

Fig.~\ref{fig:card-q-error} shows the Q-errors of the last join, whose cardinality is the query result size itself, achieved by DW and its \xbound{}-ed variant (clipping the original estimates with the corresponding lower bounds).

We observe that (i) Fabric DW underestimates 81.5\% of queries, with Q-errors reaching $5.6 \cdot 10^8$, and (ii) its \xbound{}-ed version reduces the median Q-error on these underestimates by 3.2x, and the Q-error's P90 by 35.8x (it is mostly the heavily underestimated queries that \xbound{} corrects). Since \xbound{} can only reduce underestimation, the (mild) 18.5\% overestimates of Fabric DW in this benchmark still remain.

These gains are not specific to Fabric DW. In the introductory Fig.~\ref{fig:cardinality-misestimation}, \xbound{} reduces the Q-error of underestimates by 8.38x in DuckDB and by 2.30x in PostgreSQL on \textsc{StackOverflow-CEB}. On the \textsc{JOBlight} benchmark, only DuckDB’s median Q-error could be improved, namely by 2.45x.

\subsection{Query Speedups}

Correcting cardinality underestimation directly increases estimated resources for query execution in Fabric DW---thus potentially improves query execution time. Fabric DW estimates the amount of work per query stage based on the cardinalities of the contained operators (among other). Correcting underestimates via \xbound{} translates to more CPU resources with a proportional node count rounded up to full integer---for a given query stage.

Fig.~\ref{fig:query-speedup} shows queries with execution speedups over default Fabric DW when using resource estimates based on (i) \xbound-ed cardinalities and (ii) true cardinalities. We warm up each query and then measure the time of a second query execution, limiting executions to 2h timeouts. Queries are ordered consistently with Fig.~\ref{fig:card-q-error} and can be matched by query number. A significant number of queries show no speedup from using actual cardinalities, since doing so increases total CPU resources only minimally. This is expected based on our initial observation of long-tail CPU underestimation. However, \xbound{} achieves significant speedup for vastly underestimated queries: 20.1x for \textsc{Q90} and 3.2x for \textsc{Q126}. These are queries that are significantly underestimated and it shows promise that \xbound{} may also fix the extreme CPU underestimation observed in production.

\subsection{Ablation Study}\label{subsec:ablation}

We study the effects of the several optimizations introduced throughout the past sections. We study the number of partitions and the precision of the Theta sketches on the two components of \xbound{}: the $\ell_0$ and the statistics of the heavy partition. Notably, each increase in the two values leads to an increased space overhead. The numbers for these have been already analyzed in \secref{subsec:stats-overhead}.

We show the effect of these optimizations on the Q-error achieved by \xbound{}'s lower bounds on the SO-CEB benchmark in Fig.~\ref{fig:ablation-study}. The two main observation are that the heavy partition contributes significantly to the quality of the lower bounds, and that increasing the precision of Theta sketch, especially for the $\ell_0$ statistics, has a better effect than increasing the number of partitions. Having a high-precision Theta sketch for heavy partition is important as well: changing the precision from $8$ to $12$ in this case improves the median Q-error from $28$ to $21$.

\subsection{Statistics Overhead}\label{subsec:stats-overhead}

\begin{figure}[t]
    \centering
    \includegraphics[width=1.0\linewidth]{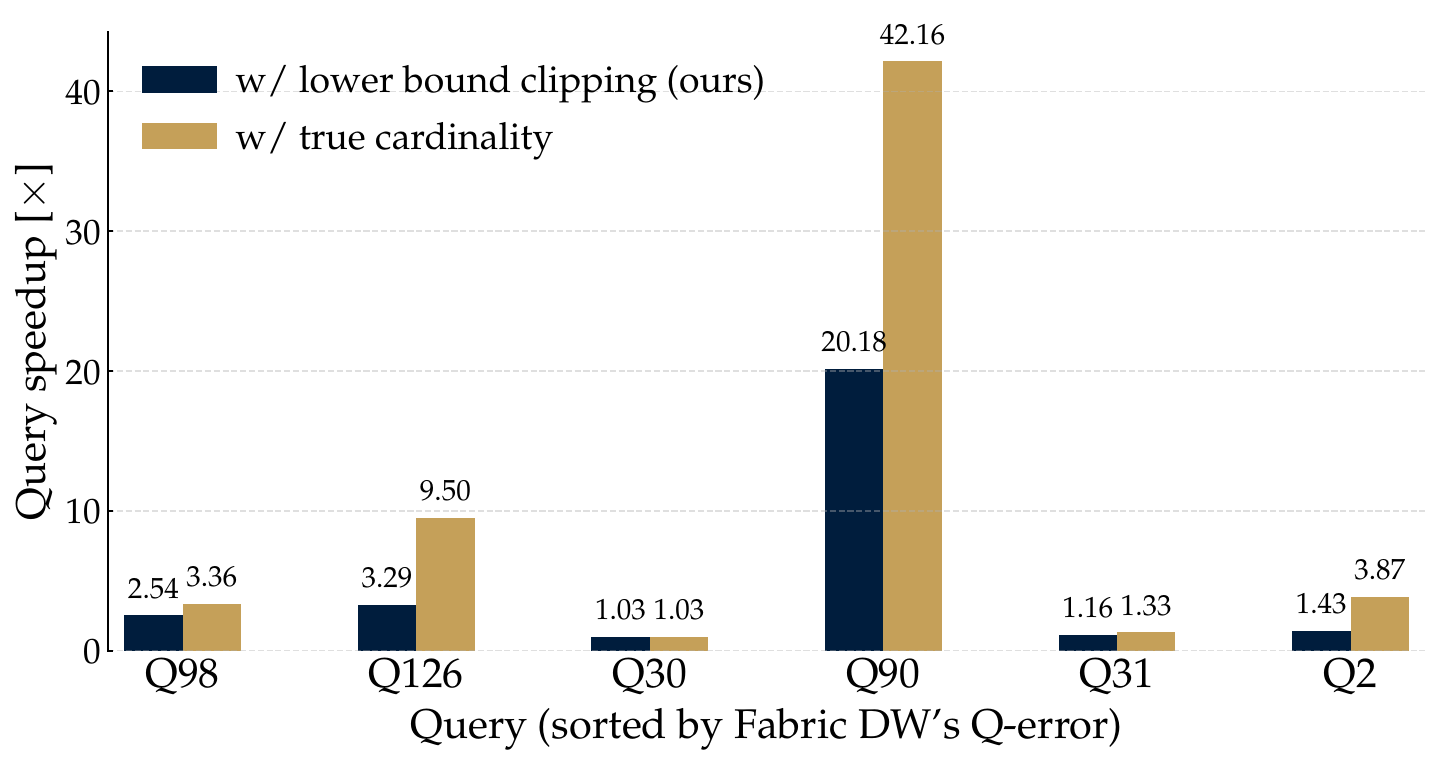}
    \caption{End-to-end query speedups on Fabric DW on the SO-CEB benchmark queries. We exclude queries with no observed speedup.}
    \label{fig:query-speedup}
\end{figure}

Depending on the configuration, the statistics used in \xbound{} can be considered lightweight. We report in Tab.~\ref{tab:stats-overhead} the Parquet sizes\footnote{Using the DuckDB Parquet writer, when enabling \texttt{zstd} compression.} and the time to compute all main statistics tables used in \xbound{}: (i) the (partitioned) $\ell_0$-statistics stored in \texttt{l0\_stats}, (ii) the $\ell_p$-norms stored in \texttt{lp\_stats}, and (iii) the statistics related to the heavy partition, in \texttt{hh\_stats}.

\sparagraph{Space \& Build.} The overhead of the partitioned $\ell_0$ statistics highly depends on the number of partitions, the configuration of the Theta sketch, and of the number of join and predicate columns; the SO-CEB benchmark features 13 join columns (and equality- and range-predicates on 11 and 16 columns, respectively).

In our case, choosing 16 partitions and a precision of 8 for \texttt{l0\_stats} results in a total size of 67.0\,MB. These statistics can be built in 6.2 min. Reducing the precision to 5, which is the minimum allowed in the DataSketches library, reduces the overhead to only 13.3\,MB. In contrast, the $\ell_p$ norms, as they do not depend on the number of partitions, take 0.2 MB. The statistics related to the \texttt{hh\_stats} can be built fast, yet, since they require high-precision sketches (see the ablation in study in \secref{subsec:ablation}), can be rather heavyweight. A precision of 12 for this set of statistics (including the support for multiple predicates, as in \secref{subsec:pred-and-heavy}) leads to 129.6 MB.

\sparagraph{Estimation Time.} Given a query on $n$ tables, we perform one Theta sketch intersection over $n$ sketches to obtain the number of distinct join keys (including a lower-bound estimate); similarly, for each of the 64 heavy keys. Our unoptimized Python prototype produces one estimate for a query in under 70 ms, with overhead largely due to \texttt{pandas}' operations on data frames (including DuckDB SQL queries to fetch the sketch-based lower bounds). We expect optimized implementations to reach <1ms estimation time.

\begin{figure}[t]
    \centering
    \includegraphics[width=1.0\linewidth]{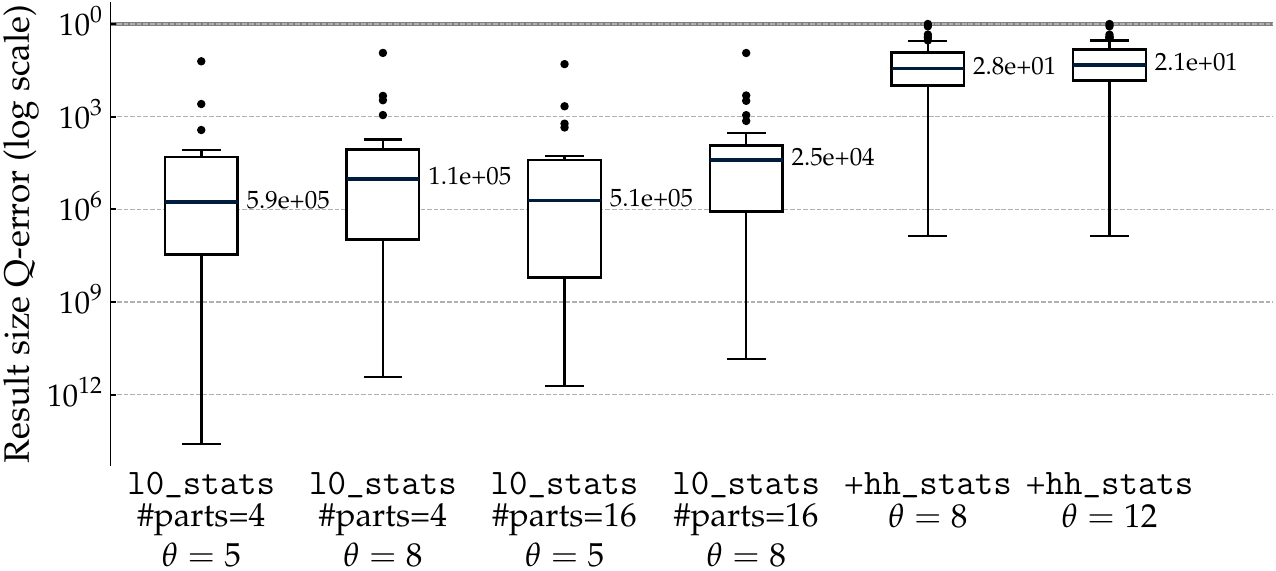}
    \caption{Quantifying the effect of \xbound{}'s components (partitioned $\ell_0$ and heavy keys statistics) and their parameters (\texttt{\#parts}: number of light partitions; $\theta$: \texttt{ThetaSketch} precision) on the Q-error of the result-size lower bounds for the supported SO-CEB queries.}
    \label{fig:ablation-study}
\end{figure}

\section{Discussion and Open Problems}\label{sec:discussion}

Having demonstrated the potential of cardinality lower bounds against the persistent \emph{underestimation} problem, we now outline avenues for generalization, hoping to inspire both practitioners and theorists to further pursue this direction.

\sparagraph{Query Types.} Currently, \xbound{} supports only acyclic queries whose joins share a single join key, a condition that inherently holds for two-way joins (e.g., the first join in a query plan) but is restrictive for multi-key joins common in practice. Moreover, while this type of acyclicity is prevalent in established benchmarks such as \textsc{JOB}~\cite{job} and \textsc{StackOverflow-CEB}~\cite{stats-ceb}, it is far less so in enterprise workloads. Extending the framework to cyclic queries, such as those in the \textsc{Subgraph-Matching} benchmark~\cite{sm-bench}, is an interesting open direction.

Another intriguing extension is to nested subqueries. While this remains open for \lpbound{}, the reverse inequalities underlying \xbound{} can be generalized to yield lower bounds on norms beyond $\ell_1$, in particular $\ell_2$, e.g., via the $\ell_4$-norms of the input vectors, enabling propagation to subsequent subexpressions.

\sparagraph{Beyond Inner Joins.} \xbound{} is able to support non-inner joins to handle a more diverse workload. The framework is flexible enough for this: for example, semi-joins $A \ltimes B$ can be supported by using the lower bound on the number of distinct join keys (\secref{subsec:exact-l0}, \secref{subsec:prob-l0}) as the prefix length when fetching the $\ell_1$ norm of $A$. For outer joins, one can always return the inner-join lower bound; tighter bounds would require information about non-joining keys, which we leave as future work.

\newcommand{\colspace}{6pt}

\begin{table}[t]
\caption{Space and build overhead (16 partitions). We vary the \texttt{ThetaSketch} precision ($\theta$) used by the statistics related to $\ell_0$ (\texttt{l0\_stats}) and the heavy partition (\texttt{hh\_stats}).}
\label{tab:stats-overhead}
\centering

\begin{tabular}{
l l c
@{\hspace{\colspace}} r
@{\hspace{\colspace}} r}
\toprule

Table & Purpose & $\theta$ & Size [MB] & Build [s] \\

\midrule

 \multirow{2}{*}{\texttt{l0\_stats}}
  & \multirow{2}{*}{Partitioned $\ell_0$-stats}
  & 5  & 13.3 & 360.1 \\
  &  & 8  & 67.0 & 373.9 \\

  \addlinespace

  \texttt{lp\_stats}
  & $\ell_p$-norms 
  & -- & 0.2 & 163.0 \\

  \addlinespace

  \multirow{2}{*}{\texttt{hh\_stats}}
  & \multirow{2}{*}{Heavy partition stats}
  & 8  & 22.6 & 18.1 \\
  &  & 12 & 129.6 & 26.5 \\
  
\bottomrule
\end{tabular}
\end{table}

Another relevant operator is group-by, for which approximate estimates exist via sketches~\cite{every-row-counts} or learning~\cite{filtered-groupby}, as well as upper bounds~\cite{lpbound}.

\sparagraph{Predicate Support.} While we support standard predicate types (equality, range, conjunctions, and disjunctions), providing good lower bounds for string predicates such as \texttt{LIKE} and \texttt{REGEX} remains open. A simple approach is sampling: the number of qualifying tuples in a sample is a (loose) lower bound, provided we do not extrapolate to the full dataset. To generate the norms needed for joins, we must lower-bound the distinct count of the join key (as with MCVs, \secref{subsubsec:equality-preds}); whether one can do better is an open question. We note that cardinality \emph{upper bounds} can be derived via tri-grams, as hinted in \safebound{} and \lpbound{}, but extending such techniques to lower bounds remains unexplored.

\sparagraph{Theory Advancements.} As the reader may have observed, \xbound{} relies on a limited set of statistics: $\{\ell_0, \ell_1, \ell_2, \ell_{\infty}\}$, $\ell_{-\infty}$, and (partitioned) min/max values, dictated by the class of inequalities we currently use. We believe it is possible to design reverse inequalities over a wider range of norms, as \lpbound{} does; any such inequality could be readily plugged into \xbound{} once the corresponding norms are available.

\sparagraph{Deteriorating Lower Bounds.} Our exact lower bounds rely on min/max statistics (zonemaps~\cite{zonemaps}), readily available in database catalogs and modern file formats. While this simplicity suits real systems, more granular statistics about the \emph{value} distribution (beyond the \emph{degree} distribution already captured by $\ell_p$-norms) could yield tighter bounds. Their role in \xbound{} is tied to the quality of lower bounds on the number of distinct join keys; ultimately, the goal is to locate joining keys more precisely, mitigating the deterioration of lower bounds as the number of joins increases.

\sparagraph{Probabilistic Lower Bounds.} To be able to run \xbound{} on large datasets, we had to resort to probabilistic lower bounds. This affected the $\ell_0$ statistics (computed at 99\% confidence) and multi-predicate support for heavy partitions. Notably, the inner-product reverse inequalities remained exact. An open question is understanding the joint guarantee when intersecting predicate-level Theta sketches for all heavy keys: while each key's degree is taken at 99\% confidence, the overall guarantee degrades as these operations are performed jointly across all heavy keys.

\sparagraph{Data Updates.} Lower bounds have an intriguing property regarding updates: inserts are supported by default, since a lower bound on $|A \Join B|$ remains valid when $\Delta_A$ is appended to $A$, unlike upper bound estimators such as \lpbound{}, where new data immediately renders degree norms stale. For incremental maintenance, we observe that $|(A + \Delta_A) \Join B| = |A \Join B| + |\Delta_A \Join B|$, as the join size is an inner product. Thus, we can reuse the previous lower bound $|A \Join B|^{-}$ and add only $|\Delta_A \Join B|^{-}$ for the new data, periodically merging deltas to limit storage and runtime overhead.
\section{Related Work}\label{sec:rel-work}

While, to the best of our knowledge, \xbound{} is the first estimator for nontrivial lower bounds on join result sizes, \emph{upper bounds} on join sizes have been extensively studied; we summarize the most relevant work at the end of this section. It is also important to distinguish our goal from the literature on fine-grained lower bounds for conjunctive queries~\cite{mengel-tut, embedding23, circuit24}. Those results show that every algorithm evaluating a given class of conjunctive queries must incur at least a certain complexity under specific computational assumptions. In other words, they yield asymptotic lower bounds, whereas our focus is on data-dependent lower bounds on join cardinalities.

\sparagraph{Learned Estimators.} Apart from the provable estimators in the last two decades, there has been a long line of research on learned cardinality estimation~\cite{joblight, naru, deepdb, factor-join}; see, for instance, the recent analysis by Heinrich et al.~\cite{heinrich-survey}. However, these estimators do not come with guarantees on their hard estimates, and can lead to arbitrarily mis-estimations when faced with complex workloads~\cite{job-complex}. An exception is FactorJoin~\cite{factor-join}, which, indeed, provides probabilistic upper bounds.

\sparagraph{Sketch-Based Estimators.} Seeing the join size as an inner product between the two corresponding degree vectors is not something new: Sketch-based estimators~\cite{skimmed-sketch, red-sketch, join-sketch, nicolau-sketch} already use this very observation. Yet, they rely on well-established sketches, which cannot guarantee lower bounds on the actual join size; they solely return an approximation, which could, itself, violate the bounds.

\sparagraph{Robust Query Processing.} The impossibility to solve the cardinality estimation problem let our community transition to a new paradigm, where the cardinality estimator is considered to be of secondary importance. This line of research started with Yannakakis' algorithm~\cite{yannakakis}, and has seen a revival with a series of Dagstuhl seminars~\cite{dagstuhl-1}, and more recently, with concrete implementations in modern database systems~\cite{pt, diamond, yannakakis+, rpt, yannakakis-eng}. 

\sparagraph{Join Size Upper Bounds.} Recall that an upper bound for a conjunctive query $Q$ is a numerical value,
which is computed in terms of statistics on the input database,
and is guaranteed to be greater or equal to the output size of any given query.
The upper bound is tight if there is a database instance, satisfying the statistics, such that the query's output size is as large as the bound.

To this end, the AGM bound~\cite{agm} is a tight upper bound that uses only the cardinalities of the relations;
in other words, it uses only the $\ell_1$-norms of full degree sequences. To upper-bound the output size, the AGM bound uses fractional edges covers and is useful for cyclic queries.

With \panda{}, Khamis, Ngo and Suciu generalize the AGM bound by also employing the $\ell_\infty$-norm~\cite{panda}.
When applied to acyclic queries, their $\texttt{max-degree}$ bound is an improvement over the AGM bound, yet, usually, is less accurate than a density-based estimator.
For instance, to upper-bound the size of $|A \Join_{X = Y} B|$, \texttt{max-degree} uses the following three inequalities (note that the first one is coming from the AGM bound):
\begin{equation*}
\begin{aligned}
  |A \Join_{X=Y} B| &\le |A| \cdot |B|,\\
  |A \Join_{X=Y} B| &\le \|\text{deg}_A(\ast\:|\:X)\|_{\infty} \cdot |B|,\\
  |A \Join_{X=Y} B| &\le |A| \cdot \|\text{deg}_B(\ast\:|\:Y)\|_{\infty}.
\end{aligned}
\end{equation*}
Our \texttt{min-degree} estimator~\secref{subsec:min-degree} takes inspiration from this observation for join size lower bounds.

\textsc{SafeBound}~\cite{safebound-1, safebound-2} uses simple, full degree sequences and computes a tight upper bound of a Berge-acyclic, full conjunctive query.
For example, if $\deg_A(\ast\:|\:X) = (a_1 \geq a_2 \geq \ldots)$ and $\deg_B(\ast\:|\:Y) = (b_1 \geq b_2 \geq \ldots)$,
then \safebound{} returns the following upper bound on the two-way join from above: $|A \Join_{X = Y} B| \leq \sum_{i} a_ib_i$. Indeed, \safebound{} can return a much better bound than the ones above.
However, its bounds are not describable by a closed-form formula; it is only given by an algorithm.

The above limitations of \safebound{} are essentially addressed by \lpbound{}~\cite{lpbound}, which, instead of relying on degree sequences, it employs their $\ell_p$-norms.
First, this eliminates the need to store degree sequences.
Second, by solving a linear program, this returns a closed-form formula that can be easily understood.
\section{Conclusion}\label{sec:conclusion}

Motivated by the pronounced long tail of CPU resource underestimates observed in production, we introduced in this work \xbound{}, a framework for computing cardinality lower bounds. Although it currently supports only inner (and outer) joins on single keys,
preliminary results show that it can correct underestimates and fix end-to-end query execution slowdown in the most painful cases on the challenging open benchmark \textsc{StackOverflow-CEB} in our system, Fabric DW.

We outlined several avenues for generalization, including extending the range of supported query types as well as advancing the underlying theory to further tighten the lower bounds. We hope to spark interest in this problem among both practitioners and theorists.

Whether our systems’ optimizers will completely overcome their Achilles heel in the decade ahead, only time will tell.

\sparagraph{Acknowledgments.} We thank our colleagues Antonio Fernández-Rodríguez, José Alberto López Magaña, and Tej Trivedi for their vital hands-on support for the Fabric Warehouse setup.

\bibliographystyle{ACM-Reference-Format}
\bibliography{xbound}

\end{document}